\documentclass[a4paper]{article}
\usepackage{color}
\usepackage{a4wide}
\usepackage{amssymb,amsfonts,latexsym,stmaryrd}
\usepackage{QED}
\usepackage[PostScript=dvips]{diagrams}
\usepackage{euscript}
\usepackage{epsfig}
\usepackage{tikz,ifdraft}
\newtheorem{theorem}{Theorem}[section]

\newtheorem{proposition}[theorem]{Proposition}

\newcommand{\beqa}{\begin{eqnarray*}}
\newcommand{\eeqa}{\end{eqnarray*}\par\noindent}

\renewcommand{\emph}[1]{\textbf{#1}}

\newcommand{\lrarr}{\longrightarrow}
\newcommand{\rarr}{\rightarrow}

\newcommand{\CC}{\mathcal{C}}
\newcommand{\labarrow}[1]{\stackrel{#1}{\longrightarrow}}

\newcommand{\Set}{\mathbf{Set}}

\newcommand{\IFF}{\; \Longleftrightarrow \;}
\newcommand{\AND}{\; \wedge \;}
\newcommand{\OR}{\; \vee \;}
\newcommand{\IMP}{\; \Rightarrow \;}

\newcommand{\da}{{\downarrow}}

\newarrow{Eq}=====
\newarrow{mon}>--->
\newarrow{rel}--+->
\newarrow{inc}C--->
\newcommand{\gl}{g_{\lambda}}

\newcommand{\natarrow}{\labarrow{\cdot}}

\newcommand{\UU}{\mathcal{U}}

\newcommand{\Real}{\mathbb{R}}

\newcommand{\SP}{\SSS}

\newcommand{\Pow}{\mathcal{P}}

\newcommand{\Nat}{\mathbb{N}}

\newcommand{\lsem}{\llbracket}
\newcommand{\rsem}{\rrbracket}

\renewcommand{\AA}{\mathcal{A}}
\newcommand{\lft}{\mathsf{l}}
\newcommand{\rgt}{\mathsf{r}}
\newcommand{\lbranch}{\lft}
\newcommand{\rbranch}{\rgt}
\newcommand{\FG}{F_{\mathsf{G}}}
\newcommand{\FSP}{F_{\mathsf{SP}}}
\renewcommand{\gl}{g_{\lft}}
\newcommand{\gr}{g_{\rgt}}
\renewcommand{\sl}{s_{\lft}}
\newcommand{\sr}{s_{\rgt}}
\renewcommand{\SP}{\mathcal{S}}
\newcommand{\GG}{\mathcal{G}}
\newcommand{\game}{\mathsf{game}}
\newcommand{\AsBc}{\mathsf{AsBc}}
\newcommand{\AcBs}{\mathsf{AcBs}}
\newcommand{\BcAs}{\mathsf{BcAs}}
\newcommand{\BsAc}{\mathsf{BsAc}}
\newcommand{\fwc}{f_{\da}}
\newcommand{\sn}{{\Downarrow}}
\newcommand{\fsn}{f_{\sn}}
\newcommand{\us}{\hat{s}}
\newcommand{\usl}{\hat{s}_{\lft}}
\newcommand{\usr}{\hat{s}_{\rgt}}
\newcommand{\PE}{\mathsf{PE}}
\newcommand{\SPE}{\mathsf{SPE}}
\newcommand{\WC}{\mathsf{WC}}
\newcommand{\SC}{\mathsf{SC}}
\newcommand{\ba}{\bar{\alpha}}
\newcommand{\alst}{\ba^{\ast}}

\newcommand{\dominates}{\succcurlyeq}

\title{Coalgebraic Analysis of Subgame-perfect Equilibria in Infinite Games without Discounting}

\author{Samson Abramsky\\Department of Computer Science, University of Oxford\\Viktor Winschel\\ 
Department of Economics, University of Mannheim}

\date{\today}

\begin{document}

\maketitle

\begin{abstract}
We present a novel coalgebraic formulation of infinite extensive games.
We define both the game trees and the strategy profiles 
by possibly infinite systems of corecursive equations.
Certain strategy profiles are proved to be subgame perfect equilibria using a novel proof principle of predicate coinduction.
We characterize all subgame perfect equilibria for the dollar auction game.
The economically interesting feature is that in order to prove these results
we do not need to rely on continuity assumptions on the payoffs which amount to discounting the future.
In particular, we prove a form of one-deviation principle without any such assumptions.
This suggests that coalgebra supports a more adequate treatment of infinite-horizon models in game theory and economics.
\end{abstract}
%
%\setcounter{tocdepth}{4} 
%\tableofcontents
%
%---------------------------------
\section{Introduction}
%---------------------------------

Infinite structures turn up at many places in economic theory. 
Prominent examples are  infinite games.
The standard approach is to employ discounted dynamic programming to derive equilibria \cite{fudenberg1991game,ljungqvist_recursive_2004}.
In this paper we  take a different approach.
We shall use the methods of \emph{coalgebra} \cite{rutten2000universal,jacobs1997tutorial}
to analyze two basic examples of infinite games without resorting to discounting and induction.

We use coalgebraic methods to define infinite extensive games and  strategy profiles, and the notion of \emph{subgame perfect equilibrium}.
We then prove that certain profiles are subgame perfect equilibria.
This leads us to formulate an (apparently) novel notion of predicate coinduction,
which is shown to be sound by reducing it to Kozen's metric coinduction \cite{kozen2007applications}.
We give a complete characterization of the subgame perfect equilibria of the dollar auction game.
We also prove a version of a standard game theoretical result, the \emph{one deviation principle},
without needing future discounting assumptions.

Our work is inspired by recent work by Pierre Lescanne \cite{lescanne2011rationality,lescanne2012backward},
but it is explicitly coalgebraic, leading to a mathematically richer and more general approach.
The coalgebraic approach we present is promising for further applications to economic modeling 
and we will discuss some directions for future work in the last section.

For a coalgebraic treatment of combinatorial games, see \cite{honsell2011conway,honsell2012categories}.
%---------------------------------
\section{Games}
%---------------------------------

We shall assume familiarity with the basic elements of coalgebra \cite{rutten2000universal,jacobs1997tutorial}.

We consider possibly infinite games of perfect information in extensive form and
fix the following sets:

\begin{itemize}
\item A set $\AA$ of \emph{agents} or \emph{players}. In our examples, we confine ourselves to two-player games, with $\AA = \{ A, B \}$.
\item A set $\CC$ of \emph{choices}. We restrict our discussion to games where a player has a choice of two options, \textit{left} or \textit{right}, at every stage in the game where it is their turn to play. Thus $\CC = \{ \lft, \rgt \}$.
\item A set $\UU$ of \emph{utility functions}, which assign a utility to each player. We  take $\UU = \Real^{\AA}$.
\end{itemize}

\noindent
The set of game trees is defined to be (the carrier of) the final coalgebra $(\GG, \gamma)$ of the functor
\[ \FG : X \; \mapsto \; \UU \, + \, \AA \times X \times X . \]
 on the category of $\Set$. The game tree is a possibly infinite binary tree. 
 The binary nodes are labelled with an agent whose turn it is to play at that stage in the game.
 The nodes have the form
\[ \langle \alpha, \gl, \gr \rangle , \]
where $\alpha$ is the agent label, and $\gl$, $\gr$ are the subgames corresponding to the left and right choices respectively.
The leaf nodes are labelled with utility functions, representing the pay off for the game.

\subsection{The $0/1$ game}
We define the  utility functions to be
\[ v := [A \mapsto 0, \; B \mapsto 1], \qquad w := [A \mapsto 1, \; B \mapsto 0] . \]
The game is defined by the following pair of mutually corecursive equations:

\[ \begin{array}{rcl}
G & = & \langle A, v, H \rangle \\
H & = & \langle B, w, G \rangle .
\end{array}
\]

\noindent
We can depict the game tree as follows:
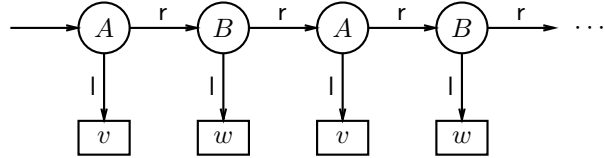
\begin{center}
\begin{tikzpicture}[scale=2.54]
% dpic version 2011.01.25 option -g for TikZ and PGF 1.01
\ifx\dpiclw\undefined\newdimen\dpiclw\fi
\global\def\dpicdraw{\draw[line width=\dpiclw]}
\global\def\dpicstop{;}
\dpiclw=0.8bp
\dpicdraw (0.132353,0) circle (0.052107in)\dpicstop
\draw (0.132353,0) node{$A$};
\dpicdraw (0.75,0) circle (0.052107in)\dpicstop
\draw (0.75,0) node{$B$};
\dpicdraw (1.367647,0) circle (0.052107in)\dpicstop
\draw (1.367647,0) node{$A$};
\dpicdraw (1.985294,0) circle (0.052107in)\dpicstop
\draw (1.985294,0) node{$B$};
\dpicdraw (0.264706,0)
 --(0.609563,0)\dpicstop
\dpicdraw (0.553627,0.013984)
 --(0.609563,0)
 --(0.553627,-0.013984)
\dpicstop
\draw (0.437134,0) node[above=-0.529412bp]{$\rbranch$};
\dpicdraw (0.882353,0)
 --(1.22721,0)\dpicstop
\dpicdraw (1.171275,0.013984)
 --(1.22721,0)
 --(1.171275,-0.013984)
\dpicstop
\draw (1.054781,0) node[above=-0.529412bp]{$\rbranch$};
\dpicdraw (1.5,0)
 --(1.844857,0)\dpicstop
\dpicdraw (1.788922,0.013984)
 --(1.844857,0)
 --(1.788922,-0.013984)
\dpicstop
\draw (1.672428,0) node[above=-0.529412bp]{$\rbranch$};
\dpicdraw (2.117647,0)
 --(2.462504,0)\dpicstop
\dpicdraw (2.406569,0.013984)
 --(2.462504,0)
 --(2.406569,-0.013984)
\dpicstop
\draw (2.290075,0) node[above=-0.529412bp]{$\rbranch$};
\dpicdraw (-0.352941,0)
 --(-0.008085,0)\dpicstop
\dpicdraw (-0.06402,0.013984)
 --(-0.008085,0)
 --(-0.06402,-0.013984)
\dpicstop
\dpicdraw (0,-0.661765) rectangle (0.264706,-0.485294)\dpicstop
\draw (0.132353,-0.573529) node{$v$};
\dpicdraw (0.617647,-0.661765) rectangle (0.882353,-0.485294)\dpicstop
\draw (0.75,-0.573529) node{$w$};
\dpicdraw (1.235294,-0.661765) rectangle (1.5,-0.485294)\dpicstop
\draw (1.367647,-0.573529) node{$v$};
\dpicdraw (1.852941,-0.661765) rectangle (2.117647,-0.485294)\dpicstop
\draw (1.985294,-0.573529) node{$w$};
\dpicdraw (0.132353,-0.132353)
 --(0.132353,-0.47721)\dpicstop
\dpicdraw (0.146337,-0.421275)
 --(0.132353,-0.47721)
 --(0.118369,-0.421275)
\dpicstop
\draw (0.132353,-0.304781) node[left=-0.529412bp]{$\lbranch$};
\dpicdraw (0.75,-0.132353)
 --(0.75,-0.47721)\dpicstop
\dpicdraw (0.763984,-0.421275)
 --(0.75,-0.47721)
 --(0.736016,-0.421275)
\dpicstop
\draw (0.75,-0.304781) node[left=-0.529412bp]{$\lbranch$};
\dpicdraw (1.367647,-0.132353)
 --(1.367647,-0.47721)\dpicstop
\dpicdraw (1.381631,-0.421275)
 --(1.367647,-0.47721)
 --(1.353663,-0.421275)
\dpicstop
\draw (1.367647,-0.304781) node[left=-0.529412bp]{$\lbranch$};
\dpicdraw (1.985294,-0.132353)
 --(1.985294,-0.47721)\dpicstop
\dpicdraw (1.999278,-0.421275)
 --(1.985294,-0.47721)
 --(1.97131,-0.421275)
\dpicstop
\draw (1.985294,-0.304781) node[left=-0.529412bp]{$\lbranch$};
\draw (2.647059,0) node{$\cdots$};
\end{tikzpicture}

\end{center}

\noindent
More formally, the above equations define a $\FG$-coalgebra $\alpha : \{ G, H \} \rarr \FG \{ G, H \}$ on the set $\{ G, H \}$. 
The $0/1$ game is given by $\lsem G \rsem$, the image of $G$ under the unique coalgebra morphism 
\[ \lsem \cdot \rsem : (\{ G, H \}, \alpha) \lrarr (\GG, \gamma) \]
from $(\{ G, H \}, \alpha)$ to the final coalgebra.

Note that we can read the utility function $v$ as `$A$ loses and $B$ wins', while $w$ corresponds to `$A$ wins and $B$ loses'. Thus the player who chooses to stop first loses.

\subsection{The dollar auction}

The $0/1$ game has only  finite states; it can be represented by a finite system of equations. The dollar auction game, a well-known example in game theory, is an infinite-state refinement, where the utility functions change as we move down the tree.

We fix a real number $r$, and define utility functions $v_n$, $w_n$ for each stage $n \in \Nat$:
\[ v_n := [A \mapsto -n, B \mapsto r -n], \qquad w_n := [A \mapsto r - n, B \mapsto -n] . \]
We define a game by simultaneous corecursion on the infinite set of variables $\{ G_n, H_n \mid n \in \Nat \}$:
\[ \begin{array}{rcl}
G_n & = & \langle A, v_n, H_n \rangle \\
H_n & = & \langle B, w_n, G_{n+1} \rangle .
\end{array}
\]
The dollar auction game is again given by $\lsem G_0 \rsem$, the image of $G_0$ under the unique coalgebra morphism into the final coalgebra.
We can depict the game tree as follows:
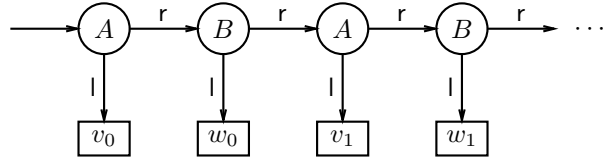
\begin{center}
\begin{tikzpicture}[scale=2.54]
% dpic version 2011.01.25 option -g for TikZ and PGF 1.01
\ifx\dpiclw\undefined\newdimen\dpiclw\fi
\global\def\dpicdraw{\draw[line width=\dpiclw]}
\global\def\dpicstop{;}
\dpiclw=0.8bp
\dpicdraw (0.132353,0) circle (0.052107in)\dpicstop
\draw (0.132353,0) node{$A$};
\dpicdraw (0.75,0) circle (0.052107in)\dpicstop
\draw (0.75,0) node{$B$};
\dpicdraw (1.367647,0) circle (0.052107in)\dpicstop
\draw (1.367647,0) node{$A$};
\dpicdraw (1.985294,0) circle (0.052107in)\dpicstop
\draw (1.985294,0) node{$B$};
\dpicdraw (0.264706,0)
 --(0.609563,0)\dpicstop
\dpicdraw (0.553627,0.013984)
 --(0.609563,0)
 --(0.553627,-0.013984)
\dpicstop
\draw (0.437134,0) node[above=-0.529412bp]{$\rbranch$};
\dpicdraw (0.882353,0)
 --(1.22721,0)\dpicstop
\dpicdraw (1.171275,0.013984)
 --(1.22721,0)
 --(1.171275,-0.013984)
\dpicstop
\draw (1.054781,0) node[above=-0.529412bp]{$\rbranch$};
\dpicdraw (1.5,0)
 --(1.844857,0)\dpicstop
\dpicdraw (1.788922,0.013984)
 --(1.844857,0)
 --(1.788922,-0.013984)
\dpicstop
\draw (1.672428,0) node[above=-0.529412bp]{$\rbranch$};
\dpicdraw (2.117647,0)
 --(2.462504,0)\dpicstop
\dpicdraw (2.406569,0.013984)
 --(2.462504,0)
 --(2.406569,-0.013984)
\dpicstop
\draw (2.290075,0) node[above=-0.529412bp]{$\rbranch$};
\dpicdraw (-0.352941,0)
 --(-0.008085,0)\dpicstop
\dpicdraw (-0.06402,0.013984)
 --(-0.008085,0)
 --(-0.06402,-0.013984)
\dpicstop
\dpicdraw (0,-0.661765) rectangle (0.264706,-0.485294)\dpicstop
\draw (0.132353,-0.573529) node{$v_0$};
\dpicdraw (0.617647,-0.661765) rectangle (0.882353,-0.485294)\dpicstop
\draw (0.75,-0.573529) node{$w_0$};
\dpicdraw (1.235294,-0.661765) rectangle (1.5,-0.485294)\dpicstop
\draw (1.367647,-0.573529) node{$v_1$};
\dpicdraw (1.852941,-0.661765) rectangle (2.117647,-0.485294)\dpicstop
\draw (1.985294,-0.573529) node{$w_1$};
\dpicdraw (0.132353,-0.132353)
 --(0.132353,-0.47721)\dpicstop
\dpicdraw (0.146337,-0.421275)
 --(0.132353,-0.47721)
 --(0.118369,-0.421275)
\dpicstop
\draw (0.132353,-0.304781) node[left=-0.529412bp]{$\lbranch$};
\dpicdraw (0.75,-0.132353)
 --(0.75,-0.47721)\dpicstop
\dpicdraw (0.763984,-0.421275)
 --(0.75,-0.47721)
 --(0.736016,-0.421275)
\dpicstop
\draw (0.75,-0.304781) node[left=-0.529412bp]{$\lbranch$};
\dpicdraw (1.367647,-0.132353)
 --(1.367647,-0.47721)\dpicstop
\dpicdraw (1.381631,-0.421275)
 --(1.367647,-0.47721)
 --(1.353663,-0.421275)
\dpicstop
\draw (1.367647,-0.304781) node[left=-0.529412bp]{$\lbranch$};
\dpicdraw (1.985294,-0.132353)
 --(1.985294,-0.47721)\dpicstop
\dpicdraw (1.999278,-0.421275)
 --(1.985294,-0.47721)
 --(1.97131,-0.421275)
\dpicstop
\draw (1.985294,-0.304781) node[left=-0.529412bp]{$\lbranch$};
\draw (2.647059,0) node{$\cdots$};
\end{tikzpicture}

\end{center}

The motivation behind this game is as follows \cite{shubik1971dollar}. 
The value of the asset being bid for is $r$; in the original example, $r$ is one dollar or 100 cents.
The asset goes to the highest bidder, who is left with a profit of $r - b$, where $b$ is the value of his highest bid. 
The loser must also pay the value of his highest bid, while getting nothing in return. 
The above definition describes the situation where each player raises their bid by one cent at each stage in the game. 
A player either gives up and finishes the game, conceding the auction to the other player and accepting their loss, or continues, 
hoping that the other player will give up. 
Both players have an incentive to continue playing well beyond the point where both will make a loss, in order to try to minimize their losses.

\section{Strategy profiles}
\label{spsec}

Intuitively, a strategy for a player $P$ of a game $G$ specifies a choice (left or right) for every node of $G$ at which it is $P$'s turn to move. 
A strategy profile specifies a strategy for every player. 
Following \cite{lescanne2011rationality}, we define the set of strategy profiles directly, as the final coalgebra $(\SP, \sigma)$  of the functor
\[ \FSP : X \; \mapsto \; \UU \, + \, \AA \times \CC \times X \times X . \]
There is an evident natural transformation $t : \FSP \natarrow \FG$ defined by projection, which induces a functor from the category of $\FSP$-coalgebras to the category of $\FG$-coalgebras. 
It sends a strategy profile $s \in \SP$ to the underlying game tree $\game(s) \in \GG$. 
We say that $s$ is a strategy profile for the game $G$ if $G = \game(s)$.

We now define a number of strategy profiles for the $0/1$ and dollar auction games.

\subsection{The $0/1$ game}

We define two strategy profiles for the $0/1$ game where $A$ always stops (plays left) and $B$ always continues (plays right), by the following simultaneous corecursion:
\[ \begin{array}{rcl}
\AsBc & = & \langle A, \lft, v, \BcAs \rangle \\
\BcAs & = & \langle B, \rgt, w, \AsBc \rangle
\end{array}
\]

\noindent
We can depict this strategy profile as follows:
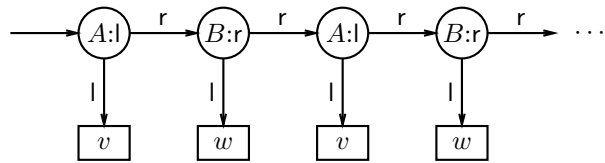
\begin{center}
\begin{tikzpicture}[scale=2.54]
% dpic version 2011.01.25 option -g for TikZ and PGF 1.01
\ifx\dpiclw\undefined\newdimen\dpiclw\fi
\global\def\dpicdraw{\draw[line width=\dpiclw]}
\global\def\dpicstop{;}
\dpiclw=0.8bp
\dpicdraw (0.132353,0) circle (0.052107in)\dpicstop
\draw (0.132353,0) node{$A{:}\lft$};
\dpicdraw (0.75,0) circle (0.052107in)\dpicstop
\draw (0.75,0) node{$B{:}\rgt$};
\dpicdraw (1.367647,0) circle (0.052107in)\dpicstop
\draw (1.367647,0) node{$A{:}\lft$};
\dpicdraw (1.985294,0) circle (0.052107in)\dpicstop
\draw (1.985294,0) node{$B{:}\rgt$};
\dpicdraw (0.264706,0)
 --(0.609563,0)\dpicstop
\dpicdraw (0.553627,0.013984)
 --(0.609563,0)
 --(0.553627,-0.013984)
\dpicstop
\draw (0.437134,0) node[above=-0.529412bp]{$\rbranch$};
\dpicdraw (0.882353,0)
 --(1.22721,0)\dpicstop
\dpicdraw (1.171275,0.013984)
 --(1.22721,0)
 --(1.171275,-0.013984)
\dpicstop
\draw (1.054781,0) node[above=-0.529412bp]{$\rbranch$};
\dpicdraw (1.5,0)
 --(1.844857,0)\dpicstop
\dpicdraw (1.788922,0.013984)
 --(1.844857,0)
 --(1.788922,-0.013984)
\dpicstop
\draw (1.672428,0) node[above=-0.529412bp]{$\rbranch$};
\dpicdraw (2.117647,0)
 --(2.462504,0)\dpicstop
\dpicdraw (2.406569,0.013984)
 --(2.462504,0)
 --(2.406569,-0.013984)
\dpicstop
\draw (2.290075,0) node[above=-0.529412bp]{$\rbranch$};
\dpicdraw (-0.352941,0)
 --(-0.008085,0)\dpicstop
\dpicdraw (-0.06402,0.013984)
 --(-0.008085,0)
 --(-0.06402,-0.013984)
\dpicstop
\dpicdraw (0,-0.661765) rectangle (0.264706,-0.485294)\dpicstop
\draw (0.132353,-0.573529) node{$v$};
\dpicdraw (0.617647,-0.661765) rectangle (0.882353,-0.485294)\dpicstop
\draw (0.75,-0.573529) node{$w$};
\dpicdraw (1.235294,-0.661765) rectangle (1.5,-0.485294)\dpicstop
\draw (1.367647,-0.573529) node{$v$};
\dpicdraw (1.852941,-0.661765) rectangle (2.117647,-0.485294)\dpicstop
\draw (1.985294,-0.573529) node{$w$};
\dpicdraw (0.132353,-0.132353)
 --(0.132353,-0.47721)\dpicstop
\dpicdraw (0.146337,-0.421275)
 --(0.132353,-0.47721)
 --(0.118369,-0.421275)
\dpicstop
\draw (0.132353,-0.304781) node[left=-0.529412bp]{$\lbranch$};
\dpicdraw (0.75,-0.132353)
 --(0.75,-0.47721)\dpicstop
\dpicdraw (0.763984,-0.421275)
 --(0.75,-0.47721)
 --(0.736016,-0.421275)
\dpicstop
\draw (0.75,-0.304781) node[left=-0.529412bp]{$\lbranch$};
\dpicdraw (1.367647,-0.132353)
 --(1.367647,-0.47721)\dpicstop
\dpicdraw (1.381631,-0.421275)
 --(1.367647,-0.47721)
 --(1.353663,-0.421275)
\dpicstop
\draw (1.367647,-0.304781) node[left=-0.529412bp]{$\lbranch$};
\dpicdraw (1.985294,-0.132353)
 --(1.985294,-0.47721)\dpicstop
\dpicdraw (1.999278,-0.421275)
 --(1.985294,-0.47721)
 --(1.97131,-0.421275)
\dpicstop
\draw (1.985294,-0.304781) node[left=-0.529412bp]{$\lbranch$};
\draw (2.647059,0) node{$\cdots$};
\end{tikzpicture}

\end{center}

\noindent
Formally, the strategy profile is given by $\lsem \AsBc \rsem$, the image of the strategy profile under the morphism to the final coalgebra.
The strategy profile $\lsem \AcBs \rsem $, where $A$ always continues and $B$ always stops, is defined symmetrically.

\subsection{The dollar auction}

In a similar fashion, we define strategy profiles for the dollar auction
by simultaneous corecursion on an infinite family of variables
$\{ \AsBc_n, \BcAs_n \mid n \in \Nat \}$:
\[ \begin{array}{rcl}
\AsBc_n & = & \langle A, \lft, v_n, \BcAs_n \rangle \\
\BcAs_n & = & \langle B, \rgt, w_n, \AsBc_{n+1} \rangle .
\end{array}
\]
Formally, $\lsem \AsBc_0 \rsem$ is a strategy profile for the dollar auction game in which $A$ always stops and $B$ always continues.
The strategy profile $\lsem \AcBs_0 \rsem$ where $A$ always continues and $B$ always stops is defined symmetrically.

\section{Subgame perfect equilibrium}
\label{spesec}

We now show, following \cite{lescanne2011rationality}, how the game-theoretic notion of \emph{subgame perfect equilibrium} can be defined coalgebraically.
Firstly, we introduce two auxiliary notions.

\subsection{Weak and strong convergence}

We introduce two predicates on strategy profiles, of \emph{weak} and \emph{strong convergence}.
A strategy profile is weakly convergent if following the choices that it specifies from the root eventually leads to a leaf; it is strongly convergent if this holds in every sub-profile.

We define weak convergence by:
\[ s \da \; \IFF \; (s = U) \OR (s = \langle P, \lft, \sl, \sr \rangle \AND \sl \da) \OR (s = \langle P, \rgt, \sl, \sr \rangle \AND \sr \da) . \]
More formally, weak convergence is an element of the powerset $\Pow(\SP)$. It is defined as the least fixpoint of the monotone function
\[ \fwc : \Pow(\SP) \lrarr \Pow(\SP) \; :: \; S \; \mapsto \; \UU \, \cup \, \{ \langle P, \lft, \sl, \sr \rangle \mid \sl \in S \} \, \cup \,
\{ \langle P, \rgt, \sl, \sr \rangle \mid \sr \in S \} .
\]

Strong convergence is defined as follows:
\[ s \sn \; \IFF \; (s = U) \OR (s = \langle P, c, \sl, \sr \rangle \AND s \da \AND \sl \sn \AND \sr \sn) . \]
More formally, it is the greatest fixpoint of the monotone function $\fsn$ defined analogously to $\fwc$
\[ \fsn : \Pow(\SP) \lrarr \Pow(\SP) \; :: \; S \; \mapsto \; \UU \, \cup \, \{ \langle P, \lft, \sl, \sr \rangle \mid \sl \in S , \, \sr \in S \} .
\]
\noindent
These fixpoints exist by the Knaster-Tarski fixed point theorem \cite{knaster1928theoreme,tarski1955lattice}.

Note that weak convergence is an \emph{inductive} notion; the profile must specify a finite path from the root to a leaf. Thus it is defined as a least fixpoint. Strong convergence, by contrast, expresses a constraint on all subtrees of an infinite tree, and hence is defined \emph{coinductively}, as a greatest fixpoint.
The examples of the strategy profiles described in the previous section are all strongly convergent, as we shall prove in Section~\ref{proofssec}.

The relationship between the two notions can be nicely characterized in  terms of coalgebraic modal logic \cite{rossiger2000coalgebras,moss1999coalgebraic}.
We can define an `always' modality $\Box$ as an operator on $\Pow(\SP)$:
\[ s \models \Box \phi \;\; \equiv \;\; s \models \phi \, \AND \, s =  \langle P, c, \sl, \sr \rangle \IMP \sl \models \Box \phi \AND \sr \models \Box \phi . \]
This operator is defined coinductively as a greatest fixpoint.
This modality generalizes straighforwardly to any polynomial functor, and can in fact be defined in a much more general way in the context of coalgebraic modal logic \cite{rossiger2000coalgebras,moss1999coalgebraic}.
If we write $\WC$ for the weak convergence predicate, and $\SC$ for the strong convergence predicate, we have the following:
\begin{proposition}
\label{scprop}
$\SC = \Box \WC$.
\end{proposition}

\subsection{Utility functions induced by strategy profiles}

A strategy profile $s$ induces a utility function $\us$.
This function is defined as follows:
%\begin{itemize}
%\item If $s = U$, then $\us = U$.
%\item If $s = \langle P, \lft, \sl, \sr \rangle$, then $\us = \usl$.
%\item If $s = \langle P, \rgt, \sl, \sr \rangle$, then $\us = \usr$.
%\end{itemize}
\[
\us =
	\left\{
		\begin{array}{ll}
			U, 	& s = U \\
			\usl,	& s = \langle P, \lft, \sl, \sr \rangle \\
			\usr,	& s = \langle P, \rgt, \sl, \sr \rangle .
		\end{array}
	\right.
\]

\noindent
In general, this function may be partial; however, if $s$ is weakly convergent, $\us$ is always a well-defined total function in $\UU$.

\subsection{Subgame perfect equilibria}

We are now ready to define the notion of a strategy profile being a \emph{subgame perfect equilibrium}.
Firstly, we define a predicate $\PE$ on strategy profiles:
\[ \PE(s) \IFF s \sn \AND (s = \langle P, \lft, \sl, \sr \rangle \IMP \usl(P) \geq \usr(P)) \AND (s = \langle P, \rgt, \sl, \sr \rangle \IMP \usr(P) \geq \usl(P)) . \]
This predicate is defined explicitly in terms of previous notions; neither induction nor coinduction is used.
It says that the choice at the root for player $P$ results in a better payoff for $P$ than the other choice would have done. Note that this predicate implies in particular that strong convergence holds.

We now define the subgame perfect equilibrium predicate $\SPE$ on $\SP$,  coinductively as a greatest fixpoint.
\[ \SPE(s) \IFF \PE(s) \AND (s = \langle P, c, \sl, \sr \rangle \IMP \SPE(\sl) \AND \SPE(\sr)) . \]
This says that the $\PE$ predicate holds at every node of the tree and we get the following analogue of Proposition~\ref{scprop}.

\begin{proposition}
$\SPE = \Box\PE$.
\end{proposition}

The next step is to show that all the strategy profiles discussed in the previous section are in fact subgame perfect equilibria. In order to do this, we will need to have an appropriate proof principle in place.
This is the topic of the next section.

\section{Predicate coinduction}

The main emphasis in coinductive proofs has been on proving equations; the main tool for this is provided by the notion of bisimulation. However, as emphasized e.g.~by Kozen \cite{kozen2007applications}, the scope of coinductive methods is broader than this. In our case, we are interested in \emph{predicates} (properties) rather than equations. In particular, we wish to show that various elements of the final coalgebra $\SP$ satisfy the $\SPE$ predicate.

We shall formulate a proof principle which is adequate to carry out these proofs. 
The principle is quite general, and applies to any $\Set$-functor $T$ which has a final coalgebra, and is equipped with a \emph{predicate lifting} \cite{jacobs2005introduction}, from which a $\Box$-modality can be defined.
Thus it applies in particular to polynomial functors
such as $\FSP$. 

Firstly, we need some notation. The final coalgebra of $T$ is denoted $(\SP, \sigma)$.
Suppose we have a $T$-coalgebra $(X, \alpha)$, which we think of as a corecursive system of equations on the set of variables $X$. We define a map
\[ \ba : \SP^X \lrarr \SP^X :: \eta \, \mapsto \, [ \, x \mapsto \sigma^{-1} \circ T \eta \circ \alpha (x) \, ] . \]
Here $\sigma^{-1} : T\SP \lrarr \SP$ is the inverse of $\sigma$, which is an isomorphism by the Lambek lemma \cite{lambek1968fixpoint}.

The following proposition follows directly by unravelling the definitions and applying the final coalgebra property:

\begin{proposition}
\label{alstprop}
The map $\ba$ has a unique fixpoint $\alst \in \SP^X$; moreover, $\alst = \lsem \cdot \rsem$, the unique coalgebra morphism from $(X, \alpha)$ to the final coalgebra.
\end{proposition}

\noindent
Now let $\phi \subseteq \SP$ be a predicate on $\SP$. 
%We lift this predicate to $\phi^X \subseteq \SP^X$:
%\[ \eta \in \phi^X \IFF \forall x \in X. \,  \eta(x) \in \phi . \]
We can  formulate our predicate coinduction principle as follows.

\begin{center}
\begin{tabular}{|c|}\hline
$\forall \eta \in \SP^X. \, \forall x \in X. \, \exists k \geq 1. \, \phi(\ba^k(\eta)(x))$ \\ \hline
$\forall x \in X. \, \Box \phi (\lsem x \rsem)$ \\ \hline
\end{tabular}
\end{center}

%\[\frac{\forall \eta \in \SP^X. \, \forall x \in X. \, \exists k \geq 1. \, \phi(\ba^k(\eta)(x)}{\forall x \in X. \, \Box \phi (\lsem x \rsem)}\]

We shall show the soundness of this principle in the following proposition.

\begin{proposition}
The predicate coinduction principle is sound.
\end{proposition}
\begin{proof}
We can consider the rule as derived by composing two more basic rules.
Firstly, from the premise, we claim that we can derive the following:
\[ \forall x \in X. \, \phi(\lsem x \rsem) . \]
This holds by taking $\eta = \lsem \cdot \rsem$, and noting that, by Proposition~\ref{alstprop}, $\ba^k(\lsem \cdot \rsem)(x) = \lsem x \rsem$.

Since the image of a coalgebra morphism is a sub-coalgebra \cite[Theorem 6.3]{rutten2000universal}, and hence an invariant \cite[Theorem 6.2.5]{jacobs2005introduction}, and since $\Box \phi$ is the largest invariant contained in $\phi$ \cite[Definition 6.3.1]{jacobs2005introduction}, the conclusion of the rule now follows.
\end{proof}

Simple as it is, this rule is useful since it allows us to derive invariants for elements of the final coalgebra which are defined by arbitrary systems of corecursive equations. We shall now apply it to the task of showing that the strategy profiles we have defined are subgame-perfect equilibria.

%------------------------------
\section{Proving properties for corecursively defined strategy profiles}
%------------------------------

\label{proofssec}

We shall now apply our predicate coinduction principle to show that the strategy profiles defined in Section~\ref{spsec} are strongly convergent and subgame perfect equilibria.

\subsection{The $0/1$-game}

We shall begin by explicitly computing the depth-three unfoldings of the corecursion variables $\AsBc$ and 
$\BcAs$. In the notation of the predicate coinduction rule, we are computing $\ba^k(\eta)(x)$ for $k=3$ and $x = \AsBc$, $x = \BcAs$. A straightforward application of the definitions yields:
\[ \AsBc^{(3)} = \langle A, \lft, v, \langle B, \rgt, w, \langle A, \lft, v, \mbox{?} \rangle \rangle \rangle , \qquad
\BcAs^{(3)} = \langle B, \rgt, w, \langle A, \lft, v, \langle B, \rgt, w, \mbox{?} \rangle \rangle \rangle . \]
Here we write {?} for $\eta(x)$, since $\eta$ is arbitrary and we have no information about this value.

A similar computation yields
\[ \AcBs^{(3)} = \langle A, \rgt, v, \langle B, \lft, w, \langle A, \rgt, v, \mbox{?} \rangle \rangle \rangle , \qquad
\BsAc^{(3)} = \langle B, \lft, w, \langle A, \rgt, v, \langle B, \lft, w, \mbox{?} \rangle \rangle \rangle . \]

\begin{proposition}
The strategy profiles $\lsem \AcBs \rsem$ and $\lsem \AsBc \rsem$ are strongly convergent.
\end{proposition}
\begin{proof}
We apply the predicate coinduction principle, with $k=2$. Using the computations of $\AcBs^{(3)}$ and $\AsBc^{(3)}$ given above, we can compute directly that $\lsem \AcBs \rsem \da$ and $\lsem \BsAc \rsem \da$, with $\widehat{\lsem \AcBs \rsem} = \widehat{\lsem \BsAc \rsem} = w$, and $\lsem \AsBc \rsem \da$ and $\lsem \BcAs \rsem \da$, with $\widehat{\lsem \AcBs \rsem} = \widehat{\lsem \BsAc \rsem} = v$. We conclude that $\lsem \AcBs \rsem \sn$ and $\lsem \AsBc \rsem \sn$, as required.
\end{proof}

\begin{proposition}
The strategy profiles $\lsem \AcBs \rsem$ and $\lsem \AsBc \rsem$ are subgame perfect equilibria.
\end{proposition}
\begin{proof}
We apply the predicate coinduction principle, with $k=3$. For $\AsBc$, we must verify that  $\AsBc^{(3)}$ and $\BcAs^{(3)}$ are $\PE$.
Using the computation of $\AsBc^{(3)}$ and  $\BcAs^{(3)}$ given above,  this reduces   to verifying the inequalities
\[ v(A) = 0 \geq 0 = v(A), \quad v(B) = 1 \geq 0 = w(B) . \]
The verification that $\lsem\AcBs \rsem$ is $\SPE$ is similar.
\end{proof}

\subsection{The Dollar Auction}

The analysis for the dollar auction will proceed along very similar lines to the $0/1$-game.

We begin by computing the depth-three unfoldings of the corecursion variables $\AsBc_n$ and 
$\BcAs_n$ for all $n \geq 0$.
\[ \AsBc_{n}^{(3)} = \langle A, \lft, v_n, \langle B, \rgt, w_n, \langle A, \lft, v_{n+1}, \mbox{?} \rangle \rangle \rangle , \qquad
\BcAs^{(3)} = \langle B, \rgt, w_n, \langle A, \lft, v_{n+1}, \langle B, \rgt, w_{n+1}, \mbox{?} \rangle \rangle \rangle . \]
A similar computation yields
\[ \AcBs_{n}^{(3)} = \langle A, \rgt, v_n, \langle B, \lft, w_n, \langle A, \rgt, v_{n+1}, \mbox{?} \rangle \rangle \rangle , \qquad
\BsAc^{(3)} = \langle B, \lft, w_n, \langle A, \rgt, v_{n+1}, \langle B, \lft, w_{n+1}, \mbox{?} \rangle \rangle \rangle . \]

The following result ca now be proved using predicate coinduction, with $k=2$, just as for the $0/1$ game.

\begin{proposition}
The strategy profiles $\lsem \AcBs_0 \rsem$ and $\lsem \AsBc_0 \rsem$ are strongly convergent.
\end{proposition}

We recall the parameter $r$ used to define the utility functions $v_n$, $w_n$.

\begin{proposition}
If $r \geq 1$, the strategy profiles $\lsem \AcBs_n \rsem$ and $\lsem \AsBc_n \rsem$ are subgame perfect equilibria for all $n$.
\end{proposition}
\begin{proof}
We apply the predicate coinduction principle, with $k=3$. For $\AsBc_n$, we must verify that  $\AsBc_n^{(3)}$ and $\BcAs_n^{(3)}$ are $\PE$ for all $n$.
Using the computations of $\AsBc_n^{(3)}$ and  $\BcAs_n^{(3)}$ given above,  this reduces   to verifying the inequalities
\[ v_n(A) = -n \geq -(n+1) = v_{n+1}(A), \quad  v_{n+1}(B) = r- (n+1) \geq -n = w_n(B) . \]
The latter inequalities are satisfied if and only if $r \geq 1$.
\end{proof}
%-------------------------------------------
\section{The one-deviation principle}
%-------------------------------------------
We now verify an important property of subgame perfect equilibria: a strategy profile is $\SPE$ if and only if it dominates any profile which differs from it in exactly one choice. In the standard game theoretical literature \cite{fudenberg_subgame-perfect_1983}, this is proved for infinite games only under strong continuity assumptions on the payoffs, which amount to discounting beyond some finite horizon, thus allowing reduction to standard backwards induction reasoning. We need no such assumptions.

Given a strategy profile $s$ of the form $\langle P, c, \sl, \sr \rangle$, we say that a profile $t$ for the same game is a one-deviation from $s$ if $t$ has one of the following forms:
\begin{itemize}
\item $\langle P, c', \sl, \sr \rangle$, where $c' \neq c$.
\item $\langle P, c, \sl', \sr \rangle$, where $\sl'$ is a one-deviation from $\sl$.
\item $\langle P, c, \sl, \sr' \rangle$, where $\sr'$ is a one-deviation from $\sr$.
\end{itemize}
This is an inductive definition.
Given a strategy profile $s$ we define  a relation $s \dominates t$, where $t$ is a one-deviation of $s$, inductively as follows:
\begin{itemize}
\item If $s = \langle P, \lft, \sl, \sr \rangle$ and $t = \langle P, \rgt, \sl, \sr \rangle$, then $s \dominates t$ iff $\usl(P) \geq \usr(P))$.
\item If $s = \langle P, \rgt, \sl, \sr \rangle$ and $t = \langle P, \lft, \sl, \sr \rangle$, then $s \dominates t$ iff $\usr(P) \geq \usl(P))$.
\item If $s = \langle P, c, \sl, \sr \rangle$ and $t = \langle P, c, \sl', \sr \rangle$, then $s \dominates t$ iff $\sl \dominates \sl'$.
\item If $s = \langle P, c, \sl, \sr \rangle$ and $t = \langle P, c, \sl, \sr' \rangle$, then $s \dominates t$ iff $\sr \dominates \sr'$.
\end{itemize}

\begin{proposition}[The one-deviation principle]
A strongly convergent strategy profile $s$ is $\SPE$ if and only if for every one-deviation $t$, $s \dominates t$.
\end{proposition}
\begin{proof}
Firstly, note that if  $\neg (s \dominates t)$ for some one-deviation $t$, this means that some subprofile of $s$ does not satisfy  $\PE$, and since $\SPE = \Box \PE$, this implies that $s$ does not satisfy $\SPE$.
For the converse, note that if for some one-deviation $t$, $s \dominates t$, this implies that the sub-profile of $s$ whose root is at the node where $t$ differs from $s$ satisfies $\PE$. If this holds for all one-deviations, then all sub-profiles of $s$ satisfy $\PE$, and hence $s$ satisfies $\SPE$.
\end{proof}
%----------------------------------------------------
\section{Complete Characterization of SPE for the Dollar Auction}
%----------------------------------------------------
We can give a \emph{complete characterization} of the subgame perfect equilibria for the dollar auction game.
\begin{theorem}
A strategy profile for the dollar auction game is SPE if and only if it is of the form 
\textit{$\AsBc$ or $A$ always stops and $B$ always continues}
or symmetrically
\textit{$\AcBs$ or $B$ always stops  and $A$ always continues}.
\end{theorem}
\begin{proof}
Firstly, note that if both players always continue from some point in the game, the profile will not be strongly convergent. Now suppose that both players choose to stop at some nodes of the game. Then there must be some node $\nu$ where player $\alpha$ chooses to stop, such that the next player choosing to stop is the other player. But then $\nu$ is not in $\PE$, since $\alpha$ could improve his payoff from that node by choosing to continue. Thus one player must always continue in any $\SPE$ profile, while the other player $\alpha$ must stop infinitely often.
Finally, $\alpha$ must in fact always choose to stop, since otherwise he could improve his payoff from any node where he chooses to continue.
\end{proof}

This analysis applies to any game sharing the following features of the dollar auction game:
\begin{enumerate}
\item At any point of the game, it is always better for a given player if the other player stops first.
\item At any point of the game, it is (strictly) better for the player who is the first to stop from that point to stop immediately rather than later.
\end{enumerate}

As a final remark on the dollar auction game, we note that Lescanne proposes to use this form of analysis to explain the \emph{rationality of infinite escalation} \cite{lescanne2011rationality}. If we don't know which strategy the other player is following, we always have an incentive to continue!
However, by our characterization result, after one round where both players choose to continue,  \emph{they both know they are not in an SPE} --- and all bets are off!
Thus it seems to us that a comprehensive analysis of escalation should use a refined model with an explicit representation of the beliefs of the players, as in Harsanyi type spaces \cite{harsanyi_games_1967} --- which can also be modelled coalgebraically \cite{moss_harsanyi_2004}.

%----------------------------
\section{Further Directions}
%----------------------------

The present paper contains what should be regarded as some very preliminary results, presented in a manner which is mainly aimed at computer scientists and mathematicians rather than economists and game theorists. Nevertheless, in our view the general idea of applying coalgebraic and other structural methods which have been developed in computer science to economics and game theory is promising, and deserves further study and development. 
In particular, many topics in economics which refer to infinite horizons and reflexivity seem tailor-made for the use of coalgebraic methods. At the same time, they can suggest new challenges and technical directions for coalgebra.

%\newpage
\bibliographystyle{plain}
\bibliography{gbib}
\end{document}